\documentclass[]{IEEEtran}
\usepackage[noadjust]{cite}
\usepackage{optidef}
\usepackage{amsmath,amssymb,amsfonts,amsthm}
\usepackage{bm}
\usepackage{algorithm}
\usepackage{algpseudocode}
\usepackage{graphicx}
\DeclareGraphicsExtensions{.eps,.pdf,.jpg}
\usepackage{textcomp}
\usepackage{xcolor}
\usepackage{balance}

\newtheorem{prop}{Proposition}
\newtheorem{lemma}{Lemma}
\newcommand{\RNum}[1]{\uppercase\expandafter{\romannumeral #1\relax}}
\def\BibTeX{{\rm B\kern-.05em{\sc i\kern-.025em b}\kern-.08em
    T\kern-.1667em\lower.7ex\hbox{E}\kern-.125emX}}

\makeatletter
\newcommand{\removelatexerror}{\let\@latex@error\@gobble}
\makeatother

\begin{document}
\title{Secrecy Throughput Maximization for Full-Duplex
Wireless Powered IoT Networks under Fairness Constraints}
%
%
%
%

\author{Roohollah~Rezaei, Sumei~Sun,~\IEEEmembership{Fellow,~IEEE,}
        Xin~Kang,~\IEEEmembership{Member,~IEEE,}
        Yong~Liang~Guan,~\IEEEmembership{Senior Member,~IEEE,}
        and~Mohammad~Reza~Pakravan,~\IEEEmembership{Member,~IEEE}
\thanks{ {\color{black}
This work was partially submitted to ICC 2019. 
R. Rezaei and S. Sun are with the Institute for Infocomm Research, Singapore
138632 (e-mail: sturr@i2r.a-star.edu.sg; sunsm@i2r.a-star.edu.sg).
R. Rezaei and MR. Pakravan are with the Electrical Engineering Department, Sharif University of Technology, Tehran, Iran (e-mail: rezaei\_roohollah@ee.sharif.edu; pakravan@sharif.edu). 
X. Kang is with the National Key Laboratory of Science and Technology
on Communications, Center for Intelligent Networking and Communications,
University of Electronic Science and Technology of China, Chengdu 610054,
China (e-mail: kangxin@uestc.edu.cn).
Y. Guan is with the School of Electrical and Electronic
Engineering, Nanyang Technological University, Singapore (e-mail: eylguan@ntu.edu.sg).}
}}

\maketitle


\begin{abstract}
In this paper, we study the secrecy throughput of a full-duplex wireless powered communication network (WPCN) for internet of things (IoT). The WPCN consists of a full-duplex multi-antenna base station (BS) and a number of sensor nodes. The BS transmits energy all the time, and each node harvests energy prior to its transmission time slot. The nodes sequentially transmit their confidential information to the BS, and the other nodes are considered as potential eavesdroppers. 
We first formulate the sum secrecy throughput optimization problem of all the nodes. The optimization variables are the duration of the time slots and the BS beamforming vectors in different time slots. The problem is shown to be non-convex. To tackle the problem, we propose a suboptimal two stage approach, referred to as sum secrecy throughput maximization (SSTM). In the first stage, the BS focuses its beamforming to blind the potential eavesdroppers (other nodes) during information transmission time slots. Then, the optimal beamforming vector in the initial non-information transmission time slot and the optimal time slots are derived. 
We then consider fairness among the nodes and propose max-min fair (MMF) and proportional fair (PLF) algorithms. The MMF algorithm maximizes the minimum secrecy throughput of the nodes, while the PLF tries to achieve a good trade-off between the sum secrecy throughput and fairness among the nodes. Through numerical simulations, we first demonstrate the superior performance of the SSTM to uniform time slotting and beamforming in different settings. Then, we show the effectiveness of the proposed fair algorithms.
\end{abstract}
\color{black}
\begin{IEEEkeywords}
wireless powered communication network (WPCN), physical layer security, Internet of Things (IoT), fairness, energy beamforming, full-duplex
\end{IEEEkeywords}

\section{Introduction}
With the advent of internet of things (IoT), many energy-constrained nodes will rely on the harvested energy to power their operation \cite{Zorzi}, especially in hazardous or difficult-to-reach locations \cite{Kang2015}. Wireless power transfer (WPT) is a promising energy harvesting solution to prolong the lifetime of these wireless networks. The two main branches of WPT are simultaneous wireless information and power transfer (SWIPT) \cite{Mishra2018} and wireless powered communication network (WPCN) \cite{Chi2017}, \cite{Zhong2018}. In SWIPT a base station (BS) transmits both energy and information towards a group of nodes, while in WPCN the nodes first harvest energy from the BS and then use the harvested energy to transmit information to the BS \cite{Zhang2013}. 

The IoT nodes either in SWIPT or WPCN transmit their information wirelessly. {\color{black}Due to the broadcast nature, the wireless medium is inherently insecure. Information transmitted over the air can be eavesdropped by the others.} Physical layer security is considered as a viable solution to ensure secure wireless information transmission in IoT networks \cite{Lee2018}. It needs no additional structure for secret key distribution and management compared to conventional cryptographic encryption algorithms. It depends on the difference of the rate between the main channel and the eavesdropper channel, referred to as secrecy rate. The artificial noise (AN) and the beamforming  are two physical layer security techniques to increase the secrecy rate \cite{Liu2018}. {\color{black}The AN is generated during the information transmission to degrade the quality of the eavesdropper channel. The beamforming increases the secrecy rate via increasing the received signal strength at the intended receiver.

Different problems in SWIPT secrecy have been studied in the literature. In WPCN scenarios, most of the previous works have considered a single information transmitter (IT). It receives energy from power station(s), called power beacon(s) (PB/PBs), and transmits confidential information to information receiver(s) (IR/IRs) \cite{Wu2016}--\cite{Bi2018}. The works in \cite{Wu2016}--\cite{Chen2017} have considered a WPCN with one PB, one IR, and one eavesdropper. They use the AN and the beamforming  to increase the secrecy rate. The authors in \cite{Guo2018} have considered a multiple ITs scenario. The ITs receive energy from a PB and transmit information to an IR in the presence of a fixed eavesdropper. In \cite{Moon2017} a helper energy harvesting node has been proposed which harvests energy in the energy transmission phase. Then it uses the harvested energy to produce the AN to disturb the eavesdroppers in the information transmission phase. In practice, it is common that in a WPCN, the PB and the IR are combined in a single node, referred to as base station (BS) \cite{Chi2017}--\cite{Zhang2013}, \cite{Lee2016}--\cite{Yang2015}. In such WPCNs, it is a serious threat  that  the legitimate nodes overhear the information of the others. Each node is an IT during its transmission time and a potential eavesdropper during its non-information transmission time. However, none of the above works has addressed this threat in these WPCNs.

To address the aforementioned problem, we consider a full-duplex WPCN which consists of one BS and a number of nodes. The BS has multiple antennas to transfer energy efficiently, while each node is equipped with a single antenna. The BS transfers energy to the nodes all the time and receives their information in sequence in the same frequency band. Hence, we refer to it as a full-duplex WPCN. Such full-duplex WPCNs have been considered in the literature (e.g., \cite{Kang2015}, \cite{Lee2016}, \cite{Lee2018a}), however, secrecy transmission has not been taken into consideration in those works. We consider that the nodes transmit their information to the BS sequentially, and each node harvests energy from the beginning up to its allocated time slot.  Then it uses the harvested energy to transmit its information to the BS in the presence of the other nodes which are considered as potential eavesdroppers. Hence, each node is an IT in its allocated time slot and an eavesdropper during the time slots of the other nodes. The full-duplex feature allows the use of the energy signal as the AN for the potential eavesdroppers. Hence, the energy signal of the BS has twofold benefits. It transfers energy to the nodes and at the same time degrades the quality of the eavesdropping channels. 
 

For the considered WPCN setup, we first look into maximizing the sum secrecy throughput of all the {\color{black}nodes}. The optimization variables are the duration of the time slots and the BS beamforming vectors. This maximization may result in unfairness among the nodes, especially when their channel conditions to the BS have large variations. To address this issue, we propose two fair algorithms to alleviate the unfairness by paying a price of reducing the sum secrecy throughput. The main contributions of this paper are summarized as follows:
\begin{itemize}
\item
We formulate the problem of maximizing the sum secrecy throughput of all the nodes in a full-duplex WPCN.  The problem is non-convex, so we propose a suboptimal two stage method, referred to as sum secrecy throughput maximization (SSTM), that converts the problem into convex optimization problems. In the first stage, we optimize the effect of the AN in the information transmission time slots via optimizing the BS beamforming vectors. Then, we substitute the solution into the original problem, which results into a convex optimization problem. The optimization variables are the BS beamforming vector in the initial non-information transmission time slot and the duration of the time slots. This problem can be solved using general convex optimization methods. However, we solve it by considering its dual problem and use the alternating optimization method to solve the dual problem. The proposed algorithm to solve the problem is much simpler and also gives intuition about the solution compared to general convex optimization methods.
\item
We consider the fairness issue among the nodes in case of maximizing the sum secrecy throughput. We propose max-min fair (MMF) and proportional fair (PLF) algorithms to decrease the difference between the secrecy throughput of different nodes.  We formulate the problems that consider these two fairness {\color{black}criteria,} and propose innovative slack variables to tackle the problems analytically. Then, we follow the same approach as the sum secrecy throughput maximization, and solve the dual problem using the alternating optimization method.      
\item
We compare the sum secrecy throughput and the fairness of the two proposed fair algorithms (MMF and PLF) with the no fairness one. We show that the MMF algorithm is much beneficial for nodes with poor channel conditions. However, it reduces the sum secrecy throughput drastically, when a node is far away from the BS. The PLF algorithm achieves a good trade-off between the sum secrecy throughput and the fairness among the nodes. The no fairness algorithm results in unfairness even between nodes with minor differences in their channel conditions.
\end{itemize}

The rest of the paper is organized as follows. In Section \RNum{2}, we describe the system model. Then, we formulate and solve the secrecy throughput maximization problem in Section \RNum{3}. We consider the unfairness among the nodes, and propose two fair algorithms in Section \RNum{4}. We evaluate our algorithm using numerical simulations in Section \RNum{5}, and demonstrate its superior to a number of benchmarking reference algorithms. Finally, we conclude the paper in Section \RNum{6}. 

\emph{Notations: We use the boldface capital and lower case to express the matrices and the vectors, respectively. $\bm{A}^H$ represents the transpose hermitian of matrix $\bm{A}$. $|.|$, $||.||$ shows the absolute value of a scalar  and the norm of a vector, respectively.}

\color{black}
\section{System Model}
We consider a full-duplex WPCN with one multi-antenna base station (BS) and $K$ single-antenna nodes. The BS has $N+1$ antennas, one of which is used for receiving information, and the rest are for transferring energy wirelessly to the nodes via beamforming. Hence, the BS also serves as the power station, called power beacon (PB), and the information receiver (IR).{\color{black}\footnote{The system model and the considered problems in this paper can also be applied to a separate PB and IR scenario.}} Each node transmits its confidential information in the presence of the other nodes which are considered as potential eavesdroppers. As a result, each node is an IT during its transmission time slot and a potential eavesdropper during the transmission time slots of the other nodes. The same frequency band is used for both the downlink energy transfer and the uplink information transmission. This will increase the bandwidth efficiency and the secrecy as the downlink BS signal acts as AN for the uplink information signals.  As shown in Fig. \ref{fig:Scenario}, we denote the downlink channel vector from the BS to node $i$ by $\bm{g_i}$, and the uplink scalar channel from node $i$ to the BS by $h_i$. The eavesdropping channel gain between node $i$ and $j$ is denoted by $h_{i,j}$. We consider a block fading channel model in which the channel state remains constant in a block of time, but may change in subsequent blocks. As we have no hidden eavesdropping nodes, we assume that the channel state information (CSI) of all links is perfectly known at the BS. 

\begin{figure}
\centering
\includegraphics[width = 1\linewidth]{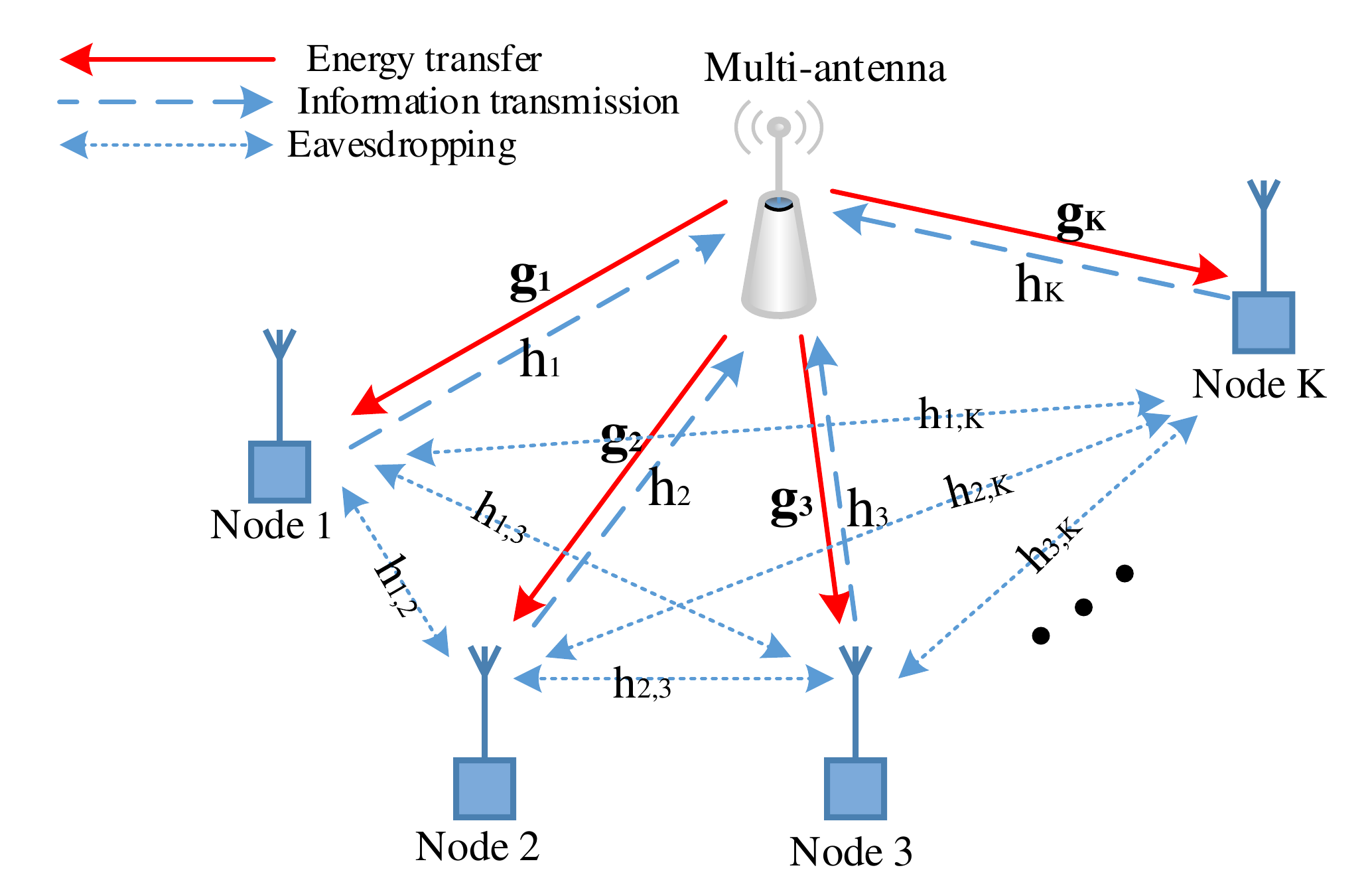}
\caption{Illustration of the considered full-duplex WPCN}
\label{fig:Scenario}
\end{figure}

\begin{figure}
\centering
\includegraphics[width = 1\linewidth]{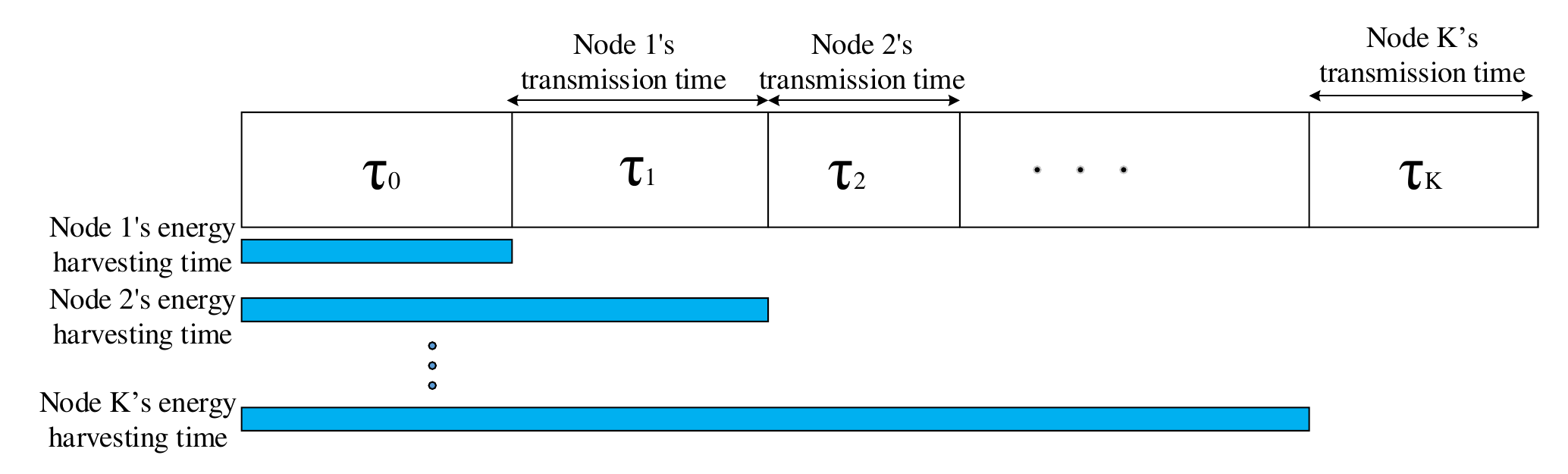}
\caption{The time frame structure of the considered full-duplex WPCN}
\label{fig:TimeFrame}
\end{figure}

{\color{black}The proposed time frame structure is depicted in Fig. \ref{fig:TimeFrame}. There are a total of $K+1$ timeslots. In time slot 0, the BS transfers energy to all nodes.  In time slot $k \in \{1,..,K\}$, the $k^{th}$ node transmits information to the BS while the BS continues transferring energy to the other nodes.} We sort the nodes according to their channel qualities such that node one has the best CSI. This will give the nodes with the worse channel quality more time to harvest more energy. The duration of the proposed time frame should be no longer than the coherence time of the channels to comply with our block fading channel model. In this paper, we focus on one single frame, so we assume a fixed CSI in our formulation. 
 As expressed in \cite{Yang2015}, when the BS has many antennas ($N \gg K$), the asymptotically optimal signal that maximizes the received energy of each node in each time slot must be in the following form 
\begin{align}
&\bm{x}_i = \sqrt{P_H}(\sum_{j=1}^K w_{i,j}\frac{\bm{g}_j}{||\bm{g}_j||}),\ \forall i \in \{0,1,...,K \}, \\
&\sum_{j=1}^K |w_{i,j}|^2 = 1, \ \forall i \in \{0,1,...,K \},
\end{align}  
\color{black} where $\bm{x}_{i}$ is the BS beamforming vector in time slot $i$, $P_H$ is the maximum power of the BS, and $w_{i,j}$ is the beamforming weight of node $j$ in time slot $i$ (i.e., $\tau_i$).  
Without loss of generality, we now focus on $\tau_i$. The received signal from the BS to node $j$ in $\tau_i$ is \color{black}
\begin{equation}
y_{j,i} = \bm{g}_j^H\bm{x}_i = \sqrt{P_H}w_{i,j}||\bm{g}_j||.
\end{equation} 

Node $i$ harvests energy from $\tau_0$ up to $\tau_{i-1}$. Hence, the received energy of node $i$ (i.e., $E_i$) can be expressed as
\begin{equation}\label{EiEquation}
E_i = \eta_i\sum_{j=0}^{i-1}\tau_j |y_{i,j}|^2 = \mu_i\eta_iP_H\sum_{j=0}^{i-1}\tau_j a_{j,i},
\end{equation}
where $\mu_i \triangleq ||\bm{g}_i||^2$, $a_{j,i}\triangleq |w_{j,i}|^2$, and $\eta_i$ is the energy conversion efficiency of node $i$. 

Node $i$ uses this energy to transfer its information during $\tau_{i}$. 
 Due to its omnidirectional information transmission, the other nodes can overhear its signal. To counteract the potential threat of information leakage, we use the BS to transmit the AN to the other nodes. The signal-to-interference-plus-noise-ratio (SINR) at node $j$, denoted by $\gamma_{i,j}$, is therefore written as
\begin{equation}\label{SINRnodes}
\gamma_{i,j} = \frac{|h_{i,j}|^2\frac{E_i}{\tau_i}}{\sigma^2+|y_{j,i}|^2}=\frac{|h_{i,j}|^2\frac{E_i}{\tau_i}}{\sigma^2+\mu_j a_{i,j}P_H}, 
\end{equation}   
$\forall j \in \{ 1,...,K \}\backslash \{ i \}$ and $\sigma^2$ is the additive white Gaussian noise (AWGN) variance at the receiver (we consider the same variance in all information receivers). The term $\mu_j a_{i,j}P_H$ is the AN from the BS. \color{black}The achievable eavesdropping rate from node $i$ to node $j$, denoted by $C_{i,j}$, is
\begin{equation}\label{EavesCapacity}
C_{i,j} = \log (1+\gamma_{i,j})=\log(1+\xi_{i,j}\frac{E_i}{\tau_i}),
\end{equation}
where $\xi_{i,j} \triangleq \frac{|h_{i,j}|^2}{\sigma^2+\mu_j a_{i,j}P_H}$.

As each node has a single antenna, it cannot harvest energy and transmit information simultaneously. So, the BS sets the beamforming weights of the transmitting node equal to zero (i.e., $a_{i,i}=0$). In addition, the BS can use self-interference cancellation (SIC) techniques to remove any low power side lobe (as mentioned its main lobes are towards other nodes) perfectly from the received information signal \cite{Kang2015}. Hence, the signal-to-noise-ratio (SNR) at the BS, denoted by $\gamma_i$, can be expressed as 
\begin{equation}
\gamma_i = \frac{|h_i|^2\frac{E_i}{\tau_i}}{\sigma^2}.
\end{equation}

The achievable rate from node $i$ to the BS can be written as
\begin{equation}\label{MainCapacity}
C_{i} = \log (1+\gamma_{i})=\log(1+\zeta_{i}\frac{E_i}{\tau_i}),
\end{equation}
where $\zeta_i \triangleq \frac{|h_i|^2}{\sigma^2}$.

Using equations \eqref{EavesCapacity} and \eqref{MainCapacity}, the worst case secrecy rate of node $i$ can be expressed as \cite{Wu2017}
\begin{equation}
R^{sec}_i = \log(1+\zeta_i\frac{E_i}{\tau_i})-\max_{j\ne i}\log(1+\xi_{i,j}\frac{E_i}{\tau_i}),\ \ \forall i \in \{1,...,K\}.
\end{equation}

The notations are provided in Table \ref{tab:Notations}.

\color{black}
\begin{table}[] 
\centering
\caption{Notations}\label{tab:Notations}
\begin{tabular}{|c|c|}
\hline
Symbols       & Definition                                                                                                    \\ \hline
$N$           & Number of the BS antennas used for energy transmission                                                           \\ \hline
$K$           & Number of the nodes                                                           \\ \hline
$x_i$         & The BS signal in time slot $i$                                                                            \\ \hline
$P_H$         & Maximum transmission power of the BS                                                                                           \\ \hline
$\sigma^2$    & Variance of the noise (same for the BS and the nodes)                                                                                         \\ \hline
$\eta_i$           & Energy conversion efficiency of node $i$                                                                                              \\ \hline
$h_i$         & Channel gain from node $i$ to the receiving antenna of the BS                                                          \\ \hline
$h_{i,j}$     & Channel gain between node $i$ and node $j$                                                      \\ \hline
$\bm{g}_i$    & The energy transmission channel between the BS and node $i$                                                         \\ \hline
$\mu_i$    & $||\bm{g}_i||^2$                                                                                              \\ \hline
$\tau_0$      & The time slot that all nodes harvest energy                                                                \\ \hline
$\tau_i$      & The time slot that node $i$ transmits information                                \\ \hline
$w_{i,j}$     & The energy beamforming weight of node $j$ in time slot $i$                                                                 \\ \hline
$a_{i,j}$ & ${|w_{i,j}|^2}$                                                                             \\ \hline
$E_{0i}$      & $\tau_0 a_{0,i}$                                                                                          \\ \hline
$E_i$         & The harvested energy of node $i$ (equation \eqref{EiEquation})                                                                                \\ \hline
$\zeta_i$         & $\frac{|h_i|^2}{\sigma^2}$                                                                                        \\ \hline
$\xi_{i,j}$         & $\frac{|h_{i,j}|^2}{\sigma^2+\mu_j a_{i,j}P_H}$                                                                                                   \\ \hline
$\xi_i$         & $\max_{j \ne i}\xi_{i,j}$                                                                                                   \\ \hline
$R_i^{sec}$   & The secrecy rate of node $i$                                                                                  \\ \hline
$D_i^{sec}(\frac{E_i}{\tau_i})$   & The secrecy throughput of node $i$ (equation \eqref{eq:DefineThrSec})                                                                                  \\ \hline
$\mathcal{B}_i\big(\frac{E_i}{\tau_i} \big)$         & $\mu_i\eta_iP_H \bigg(\frac{ \zeta_i-\xi_i}{(1+\zeta_i\frac{E_i}{\tau_i})(1+\xi_i\frac{E_i}{\tau_i})} \bigg)$                                                                                            \\ \hline
$\nu$     & Dual variable                                                                                                \\ \hline
\end{tabular}
\end{table}

\section{Sum Secrecy Throughput Maximization}\label{seq:SectionMaxSumSecrecy}
We maximize the sum secrecy throughput of all the nodes in one single time frame. 
\color{black}
We formulate {the problem of the sum} secrecy throughput as 
\begin{maxi!}
{\bm{\tau},\bm{A},\bm{a}_0}{\sum_{i=1}^K \tau_i\bigg( \log{\big(1+\frac{ \zeta_i E_i}{\tau_i}\big)} -   \max_{j\ne i}\log{\big(1+\frac{\xi_{i,j} E_i }{\tau_i}\big)}\bigg),}{\label{PMainSecrecy}}{}
  \addConstraint{ \sum_{j=1}^K a_{i,j}= 1, \  \forall i \in \{0,1,...,K\}} \label{equ:MAXRECVPTRANAVG}
\addConstraint{ \sum_{i=0}^K \tau_i = 1  }
\addConstraint{a_{i,j} \ge 0,\ \forall i \in \{0,1,...,K \},\ \forall j \in \{ 1,...,K \}} 
\addConstraint{ \tau_i \ge 0,\ \forall i \in \{0,1,...,K \},}
\end{maxi!}
where $\bm{\tau}\triangleq (\tau_0,\tau_1,...,\tau_K)$, $\bm{a}_0 \triangleq (a_{0,1},...,a_{0,K})$, $\bm{A} \triangleq [a_{i,j}]_{i=1:K}^{j=1:K}$, and the total time frame duration is normalized to one. The vector $\bm{a}_0$ expresses the beamforming weights during $\tau_0$, and the matrix $\bm{A}$ represents the beamforming vectors in the information transmission time slots ($\tau_1$ up to $\tau_K$). Equation \eqref{equ:MAXRECVPTRANAVG} expresses $K+1$ constraints with respect to each time slot. {The problem} defined by equation \eqref{PMainSecrecy}  is non-convex. Its objective function has maximization, and some optimization parameters are in the denominator of a logarithmic function. We propose a suboptimal approach, named as sum secrecy throughput maximization (SSTM), that converts this problem into convex optimization problems by the following two stages:
\begin{enumerate}
\item
\textbf{Blinding the non-transmitting nodes}, we optimize the beamforming vectors in $\tau_1$ up to $\tau_K$ (matrix $\bm{A}$)  to blind the non-transmitting nodes. As previously mentioned, these nodes are considered as potential eavesdroppers. 
\item
\textbf{Optimizing the initial beamforming vector and the time slots}, we put the obtained optimized beamforming vectors in the previous stage into the original problem. {\color{black}Then, the optimal of the beamforming vector in $\tau_0$ (vector $\bm{a}_0$) and the optimal time slots ($\bm{\tau}$) to maximize the sum secrecy throughput are derived.}
\end{enumerate}
\subsection{Blinding the non-transmitting nodes}
Each node transmits its information during its allocated time slot in the presence of the other nodes. In this stage, we optimize the beamforming vector in the information transmission time slot of each node to maximize the intereference of the non-transmitting nodes. To this end, we have to minimize $\xi_{i}\overset{\Delta}{=} \max_{j\ne i}\xi_{i,j}$  for each time slot $\tau_{i}$ as
\begin{mini!}|l|
 {{\bm{a}_i}}{\xi_i = \max_{j\ne i}\xi_{i,j},}{\label{minEaves}}{}
\addConstraint{\sum_{j=1}^K a_{i,j}=1}
\addConstraint{a_{i,j}\ge 0,\ \forall j \in \{ 1,...,K \}},
 \end{mini!}
$\forall i \in\{1,...,K\}$. $\bm{a}_i \triangleq \bm{A}(i,:) = (a_{i,1},...,a_{i,K})$ is the $i^{th}$ row of the Matrix $\bm{A}$, which is the beamforming vector during $\tau_i$. The parallel subproblems defined by equation \eqref{minEaves} can be solved independently for each $i\in \{ 1,...,K \}$. As it can be seen from equation \eqref{SINRnodes}, the subproblem $i$ tries to find the beamforming weights to minimize the maximum ratio of the eavesdropping channel gain and the interference plus noise of all the non-transmitting nodes (all the nodes except node $i$). This problem is equivalent to the maximization of the minimum interference. Hence, the primary goal of the beamforming in the information transmission period is to optimize the AN.
 
\begin{prop}\label{prop:Blinding}
Algorithm \ref{alg:BWeights} yields the optimal solution of minimizing $\xi_{i,j}$ (problem defined by equation (\ref{minEaves})). This algorithm terminates at most in $K-2$ steps.
\end{prop}

\begin{figure}
\removelatexerror
\begin{algorithm}[H]
\caption{Optimal beamforming weights in time slot $i$}\label{alg:BWeights}
\begin{algorithmic}[1]
\State \textbf{Initialization:} $flag = 1$, $B = \{1,...,K\} \backslash \{ i \} , \ a_{i,i}=0$
\While{($flag$)}
\State $flag = 0$
\For{$ j \in B$}
\State $a_{i,j}=\frac{1 + \frac{\sigma^2}{P_H} \big( \sum_{l \in B} \frac{1}{\mu_l}-\frac{1}{|h_{i,j}|^2}\sum_{l \in B} \frac{|h_{i,l}|^2}{\mu_l} \big)}{\frac{\mu_j}{|h_{i,j}|^2}\sum_{l \in B} \frac{|h_{i,l}|^2}{\mu_l}}$  \label{step:MainOneAlgOne}
\If{$a_{i,j}<0$}
\State $a_{i,j} = 0$, $flag = 1$, $B = B \backslash \{ j \}$
\State break;
\EndIf
\EndFor
\EndWhile
\end{algorithmic}
\end{algorithm}
\end{figure}
\begin{proof}
The proof is given in appendix \ref{PropositionBeta}
\end{proof}

Algorithm \ref{alg:BWeights} (step \ref{step:MainOneAlgOne}) reveals some interesting facts about the beamforming weights. If the channel between node $i$ (as the information transmitter) and node $j$ (as a potential eavesdropper) ($|h_{i,j}|^2$) is strong, the beamforming weight for node $j$ needs to be increased. On the other hand, if this channel is much degraded, its weight becomes zero. This means that the BS does not need to jam this node. Another interesting point is that if an eavesdropper is near to the BS, its weight decreases. It is reasonable as the BS needs less power to generate the same AN for this node than the nodes that are farther away from the BS.

\subsection{Optimizing the initial beamforming vector and the time slots}
We use the obtained optimal beamforming vectors (matrix $\bm{A}$) in the previous stage to simplify the original problem (equation \eqref{PMainSecrecy}) as

\begin{maxi!}
{\bm{\tau},\bm{a}_0}{\sum_{i=1}^K \tau_i\bigg( \log{(1+\frac{ \zeta_i E_i}{\tau_i})} - \log{(1+\frac{\xi_{i} E_i }{\tau_i})}\bigg),}{\label{PMain}}{}
  \addConstraint{ \sum_{j=1}^K a_{0,j}= 1 } \label{equ:Subslot0}
\addConstraint{ \sum_{i=0}^K \tau_i = 1}\label{equ:TauIs}
\addConstraint{ a_{0,j} \ge 0,\ \forall j \in \{ 1,...,K \}}
\addConstraint{ \tau_i \ge 0,\ \forall i \in \{0,1,...,K \}.}
 \end{maxi!}
 Due to the $\tau_0a_{0,i}$ term in $E_i = \mu_i\eta_iP_H(\tau_0a_{0,i} + \sum_{j=1}^{i-1}\tau_ja_{j,i}),\ \forall i \in \{1,...,K \}$, we have the product of two variables inside the $\log$ part of the objective function. As a result, it can be shown that the objective function is non-convex in general (consider the non-convexity of $\log(1+\tau_0a_{0,1})$ with respect to $\tau_0$ and $a_{0,1}$). To resolve this issue, we define $E_{0i} \triangleq \tau_0a_{0,i}$ and rewrite $E_i = \mu_i\eta_iP_H(E_{0i} + \sum_{j=1}^{i-1}\tau_ja_{j,i}),\ \forall i \in \{1,...,K \}$. We reformulate the problem as
\begin{maxi!}
{\bm{E}_{0},\bm{\tau}}{\sum_{i=1}^K \tau_i\bigg( \log{\big(1+\frac{ \zeta_i E_i}{\tau_i}\big)} -  \log{\big(1+\frac{\xi_{i} E_i }{\tau_i}\big)}\bigg),}{\label{PSeveralnode}}{}
\addConstraint{\sum_{i=1}^K (E_{0,i}+\tau_i) = 1} \label{equ:E0i}
\addConstraint{E_{0i},\tau_i \ge 0, \ \forall i \in \{ 1,...,K\}, }
\end{maxi!}
where $\bm{E}_{0} \triangleq ( E_{01},E_{02},...,E_{0K} )$, and we combined equations \eqref{equ:Subslot0} and \eqref{equ:TauIs}. We omitted $\tau_0$ from the formulation, but after the optimization, we can obtain it from equation \eqref{equ:TauIs} (As $E_{0i}$'s $\ge0$, $\forall i \in \{ 1,...,K\}$, we have $\tau_0 \ge 0$). We prove in the following lemma that the throughput of each node, defined as 
\begin{equation}
D_i^{sec}(\frac{E_i}{\tau_i}) \triangleq \tau_i\bigg( \log{\big(1+\frac{ \zeta_i E_i}{\tau_i}\big)} - \log{\big(1+\frac{\xi_{i} E_i }{\tau_i}\big)}\bigg), \label{eq:DefineThrSec}
\end{equation}
 is concave. Therefore, the objective function is concave as it is a non-negative sum of concave functions \cite{Boyd2004}. In addition, the equality constraint (equation \eqref{equ:E0i}) is linear, so the above problem is convex.
\begin{lemma}\label{lem:ConcaveDi}
The throughput of each node $(D_i^{sec}(\frac{E_i}{\tau_i}))$ is a concave function of $\bm{E_0}$ and $\bm{\tau}$. 
\end{lemma} 
\begin{proof}
The proof is given in appendix \ref{app:LemmaConcave}
\end{proof}

As the above problem is convex, we could solve it using interior point methods, but we propose another method that needs less complexity and gives us intuition about the solution. It is easy to verify that the Slater condition holds for the problem, and due to the convexity, the strong duality holds. Hence, we can obtain the optimal solution via solving the dual problem, which can be written as 
\begin{maxi!}
{\nu}{\min_{\bm{\tau},\bm{E}_{0}}\ \mathcal{L}(\bm{\tau},\bm{E}_{0},\nu),}{\label{DualProb}}{}
\addConstraint{\tau_i,E_{0i}\ge 0,\ \forall i \in \{ 1,...,K \},} 
\end{maxi!}
where $\mathcal{L}(\bm{\tau},\bm{E}_{0},\nu)=\sum_{i=1}^K -D_i^{sec}(\frac{E_i}{\tau_i})+ \nu(\sum_{i=1}^K (E_{0i} + \tau_i)  -1).$ In order to solve the above problem, we use the alternating optimization method. We assume a fixed $\bm{\tau}$ and obtain the optimal $\bm{E}_0$ and $\nu$ via Algorithm \ref{alg:FixTauE0}. This algorithm uses just a simple bisection method to obtain the optimal values. We express in Proposition \ref{theo:E0Opt} that this algorithm yields the optimal values. Then, using the gradient descent method, we update $\bm{\tau}$ and obtain the new corresponding optimal $\bm{E}_0$ and $\nu$. We repeat this iterative method, which is summarized in Algorithm \ref{alg:TauE0}, to find the optimal solution. Due to the convexity of the above problem, we can obtain the optimal solution ($\bm{\tau}^*,\bm{E}_{0}^*$) via this iterative method. Using this approach, we obtain half of the optimization variables ($\bm{E}_0^*$) via a simple bisection method, and obtain the other half using a fast convergent method. It should be noted that for some channel realizations, the optimal $E_{0i}$'s lie in the boundaries of the convex region, and finding the optimal solution in these cases is much easier with the proposed algorithm. We express the following proposition which gives us the optimal $\bm{E}_0$ and $\nu$ for a fixed $\bm{\tau}$.  
\begin{prop} \label{theo:E0Opt}
Algorithm \ref{alg:FixTauE0} yields the optimal solution of the problem defined by equation \eqref{DualProb} for a fixed $\tau$.
\end{prop}
\begin{proof}
The proof is given in appendix \ref{Proposition Lambda}.
\end{proof}

\begin{figure}
\removelatexerror
\begin{algorithm}[H]
\caption{optimal $\bm{E}_{0}$ and $\nu$}\label{alg:FixTauE0}
\begin{algorithmic}[1]
\State \textbf{Initialization:} $B=\{1,...,K\}$, $\nu_{min} = 0, \nu_{max} = \max_{i} \mu_i\eta_iP_H(\zeta_i-\xi_i) $
\While{$Er(\nu) > \epsilon$}
\State $\nu = \frac{\nu_{min}+\nu_{max}}{2}$
\For{$i \in B$}
\State $E_i = \tau_i \frac{-(\zeta_i +\xi_i)+\sqrt{(\zeta_i -\xi_i)^2+
\frac{4\zeta_i \xi_i \mu_i\eta_i P_H(\zeta_i -\xi_i)}{\nu}}}{2\zeta_i \xi_i}$ \label{step:optimalEnergy}
\State $E_{0i} = \max\{0,\frac{E_i}{\mu_i\eta_i P_H}-\sum_{j=1}^{i-1}\tau_j a_{j,i}\}$\label{step:NonNeededE0}
\EndFor
\State $Er(\nu) \triangleq \sum_{i=1}^K (E_{0i}+\tau_i) - 1$ \label{Alg:defErorNoFair}
\If{$Er(\nu)>0$}
\State $\nu_{min} = \nu $
\Else 
\State $\nu_{max} = \nu$
\EndIf
\EndWhile
\State $E_i = \mu_i\eta_i P_H(E_{0i}+\sum_{j=1}^{i-1}\tau_j a_{j,i}), \ \forall i \in \{1,...,K\}$
\end{algorithmic}
\end{algorithm}
\end{figure}
The intuition behind Algorithm \ref{alg:FixTauE0} is that for a fixed $\bm{\tau}$, it derives the optimal energy (${E_i}$ in step \ref{step:optimalEnergy}) that each node must have in order to maximize the sum secrecy throughput. As the harvested energy of node $i$ during ($\tau_1$,...,$\tau_{i-1}$) is from the AN and fixed for this algorithm, it can only reach the optimal $E_i$ via ${E_{0i}}$ (step \ref{step:NonNeededE0}). The available energy in $\tau_0$ is limited, so the algorithm can reach the optimal $E_i$'s and the corresponding maximum secrecy rate to some extent ($Er(v)$ defined in step \ref{Alg:defErorNoFair}  is actually $\sum_{i=1}^K E_{0i} -\tau_0 $, which expresses the limitation). We can realize from the algorithm (step \ref{step:NonNeededE0}) that the nodes which receive enough energy using the AN, do not receive energy during $\tau_0$. {\color{black}Instead, the algorithm tries to transfer energy to nodes that do not receive their optimal energy during the AN period.}  

We derive the closed-form gradient vector of $\mathcal{L}$ with respect to $\bm{\tau}$ as
\begin{align}
\nabla_{\bm{\tau}}\mathcal{L} &= -\big(\log (1+\zeta_i \frac{E_i}{\tau_i})+\frac{1}{1+\zeta_i \frac{E_i}{\tau_i}}\nonumber\\& - \log (1+\xi_i \frac{E_i}{\tau_i})-\frac{1}{1+\xi_i \frac{E_i}{\tau_i}}\big)+\nu.  \label{equ:Gradient}
\end{align}

We use the gradient descent method to update $\bm{\tau}$, and choose its step using the backtracking line search, which converges quickly. In the line search, we use $\Delta \bm{\tau}= -\nabla \mathcal{L}$, $\alpha = 0.5$, and $\beta = 0.5$ \cite{Boyd2004}. 
Algorithm \ref{alg:TauE0} expresses the complete steps to obtain the optimal solution of the problem.

\begin{figure}
\removelatexerror
\begin{algorithm}[H] 
\caption{optimal $\bm{E}_{0}$ and $\bm{\tau}$}\label{alg:TauE0}
\begin{algorithmic}[1]
\State \textbf{Initialization:} $l=0$, $\alpha = 0.5$, $\beta = 0.5$ , $\tau_{i}^{(1)} = \frac{1}{K+1},\ \forall i \in \{ 0,...,K\}$
\Repeat 
\State $l = l + 1$
\State $t= 1$, obtain $\bm{E_{0}^{(l)}}$ and $\nu^{(l)}$ using algorithm \ref{alg:FixTauE0}
\State calculate $\nabla_{\bm{\tau}}^{(l)}\mathcal{L}(\bm{\tau}^{(l)},\bm{E}_0^{(l)},\nu^{(l)})$ using equation \eqref{equ:Gradient} 
\Repeat
\State $\bm{\tau}^{(l+1)} = \bm{\tau}^{(l)} - t\nabla_{\bm{\tau}}^{(l)}\mathcal{L}$
 \State obtain $\bm{E_{0}}^{(l+1)}$ and $\nu^{(l+1)}$ using algorithm \ref{alg:FixTauE0}
 \State $t = \beta t$ 
\Until{($\mathcal{L}^{(l+1)}>\mathcal{L}^{(l)}-\alpha t || \nabla_{\bm{\tau}}^{(l)}\mathcal{L}||^2$)}
\Until{(($|\mathcal{L}^{(l+1)}-\mathcal{L}^{(l)}|>\epsilon$) and ($||\bm{\tau}^{(l+1)}-\bm{\tau}^{(l)} ||>\epsilon$))} 
\end{algorithmic}
\end{algorithm}
\end{figure}

\section{Considering Fairness Among the Nodes}
In WPCN, nodes farther away from the BS receive less energy than the nearer ones and hence, need more energy to transmit information. This phenomenon is known as the double-near-far \cite{Guo2018} problem, and results in unfairness among the secrecy throughput of the nodes when the sum secrecy throughput is maximized.  To alleviate this issue, we propose max-min fair (MMF) and proportional fair (PLF) algorithms based on max-min and proportional fairness, respectively. The well-known max-min fairness maximizes the minimum secrecy throughput of all the nodes, which usually leads to similar throughput for all the nodes. However, it can decrease the sum secrecy throughput drastically when a node is far away from the BS. {\color{black} The proportional fairness tries to achieve a good trade-off between  the sum secrecy throughput and the fairness among the nodes.} 
In the following two subsections, we formulate and obtain the solutions for the max-min fairness and the proportional fairness. Similar to Section \ref{seq:SectionMaxSumSecrecy}, we use the two stage approach to convert the problems into tractable ones. Since blinding the non-transmitting nodes may have a negligible effect on the fairness among the nodes, we consider the same first stage as in Section \ref{seq:SectionMaxSumSecrecy}, and in the following two subsections we only consider the second stage.  
\subsection{Max-min Fairness} \label{SubSec:MaxMin}
The max-min fairness maximizes the minimum secrecy throughput of all the nodes, and it can be formulated as
\begin{maxi!}
{\bm{E}_{0},\bm{\tau}}{\min_{i}D_i^{sec}(\frac{E_i}{\tau_i}),}{}{}
\addConstraint{\sum_{i=1}^K (E_{0i}+\tau_i) = 1} \label{}
\addConstraint{{E}_{0i},\tau_i \ge 0,  \ \forall i \in \{ 1,...,K\}.}
\end{maxi!}
The difference between this problem and the problem defined by equation \eqref{PSeveralnode} (no fairness case) is that it maximizes the minimum secrecy throughput instead of the sum secrecy throughput. This problem is convex, and can be solved using existing convex optimization methods. Instead, similar to the no fairness case, we propose an algorithm that has less computational complexity and gives us intuition about the solution. We introduce a new slack variable, denoted by $\phi$, and reformulate the above problem as  
\begin{maxi!}
{\phi,\bm{E}_{0},\bm{\tau}}{\phi,}{\label{opt:SlackAdd}}{}
\addConstraint{D_i^{sec}(\frac{E_i}{\tau_i}) \ge \phi, \ \forall i \in \{ 1,...,K\} } \label{eq:SlackConst}
\addConstraint{\sum_{i=1}^K (E_{0i}+\tau_i) =1} \label{}
\addConstraint{{E}_{0i},\tau_i \ge 0,  \ \forall i \in \{ 1,...,K\}.}
\end{maxi!}
The slack variable in fact represents the minimum throughput of all the nodes. Similar to the no fairness case, we consider the dual problem, which can be represented as
\begin{maxi!}
{\bm{\lambda},\nu}{\min_{\phi,\bm{E}_{0},\bm{\tau}}\ \mathcal{L}(\phi,\bm{E}_{0},\bm{\tau},\bm{\lambda},\nu),}{\label{equ:DualMaxMin}}{}
\addConstraint{\lambda_i,E_{0i},\tau_i \ge 0,\ \forall i \in \{ 1,...,K \},}
\end{maxi!}
where $\bm{\lambda} \triangleq (\lambda_1,...,\lambda_K)$ and 
\begin{align*}
\mathcal{L}({\phi},\bm{E}_{0},\bm{\tau},\bm{\lambda},\nu)&= -\phi + \sum_{i=1}^K\lambda_i \big(\phi - D_i^{sec}(\frac{E_i}{\tau_i})  \big) \\&+ \nu(\sum_{i=1}^K (E_{0i} + \tau_i)  -1).
\end{align*}

We use the alternating optimization method to solve the dual problem. 
 We assume a fixed $\bm{\tau}$ and optimize the other variables via Algorithm \ref{alg:FixTauE0MaxMin}. The optimality of Algorithm \ref{alg:FixTauE0MaxMin} is expressed in Proposition \ref{theo:E0OptMaxMin}. Then, we update $\bm{\tau}$ using the gradient descent method and again optimize the other variables for the updated $\bm{\tau}$. We repeat this procedure to obtain the optimal solution. The algorithm for this part is similar to the no fairness case, and we omit it for brevity. The max-min problem is more complex than the no fairness problem and has more optimization variables for a fixed $\bm{\tau}$. In the following proposition, we prove that Algorithm \ref{alg:FixTauE0MaxMin} yields the optimal solution for a fixed $\bm{\tau}$.  
\begin{prop} \label{theo:E0OptMaxMin}
Algorithm \ref{alg:FixTauE0MaxMin} yields the optimal solution of the problem defined by equation \eqref{equ:DualMaxMin} for a fixed $\bm{\tau}$. As noted in Table \ref{tab:Notations}, $\mathcal{B}_i(\frac{E_i}{\tau_i})\triangleq \mu_i\eta_iP_H \bigg(\frac{ \zeta_i-\xi_i}{(1+\zeta_i\frac{E_i}{\tau_i})(1+\xi_i\frac{E_i}{\tau_i})} \bigg)$. 
\end{prop}
\begin{proof}
The proof is given in appendix \ref{sec:PropositionLambdaMaxMin}.
\end{proof}

\begin{figure}
\removelatexerror
\begin{algorithm}[H]
\caption{optimal $\bm{E}_{0}$, $\bm{\lambda}$, $\phi$ and $\nu$ in the max-min fairness case}\label{alg:FixTauE0MaxMin}
\begin{algorithmic}[1]
\State \textbf{Initialization:} $\phi_{min} = 0,\ \phi_{max} = \min_{i} \tau_i(\log \zeta_i-\log\xi_i),\ \bm{\lambda}=(0,...,0) $, $flag = 1$ 
\While{$Er(\phi) > \epsilon$} 
\State $\phi = \frac{\phi_{min}+\phi_{max}}{2}$, $B_\phi=\{1,...,K\}$ \label{step:MinimumThr}
\While{($flag$)}	\label{step:BeginingBPhi}
\State $flag = 0$
\For{$j \in B_\phi$}
\State $E_j = \frac{\tau_j(e^{\frac{\phi}{\tau_j}}-1)}{\zeta_j-\xi_j e^{\frac{\phi}{\tau_j}}}$ \label{alg:Step7Ej}
\State $E_{0j} = \frac{E_j}{\mu_j\eta_j P_H}-\sum_{l=1}^{j-1}\tau_l a_{l,j}$\label{step:sumchannel} \label{alg:Step8E0j}
\If{$E_{0j}<0$}
\State $flag = 1$, $ E_{0j} = 0, \ B_\phi=B_\phi\backslash \{j\}$,\ 
\State break
\EndIf  \label{alg:StepCheck}
\EndFor
\EndWhile	\label{step:EndBPhi}
\State $Er(\phi) \triangleq \sum_{i=1}^K (E_{0i}+\tau_i) - 1$\label{step:E0Available}
\If{$Er(\phi)<0$}
\State $\phi_{min} = \phi $
\Else 
\State $\phi_{max} = \phi$
\EndIf
\EndWhile 
\State $E_i = \mu_i\eta_i P_H(E_{0i}+\sum_{j=1}^{i-1}\tau_j a_{j,i}),\ \forall i \in \{ 1,...,K \}$
\State $\nu = \frac{1}{\sum_{i \in B_\phi}\frac{1}{\mathcal{B}_i(\frac{E_i}{\tau_i})}}$
\For{$i \in B_{\phi}$}
\State $\lambda_i = \frac{\nu}{\mathcal{B}_i(\frac{E_i}{\tau_i})}$
\EndFor
\end{algorithmic}
\end{algorithm}
\end{figure}

{\color{black}Algorithm \ref{alg:FixTauE0MaxMin} considers an initial interval for $\phi$ and sets $\phi=\frac{\phi_{min}+\phi_{max}}{2}$ in step \ref{step:MinimumThr} in each iteration. As previously mentioned, $\phi$ represents the minimum throughput of all the nodes, so the throughput of each node must be at least $\phi$. The needed energy to reach $\phi$ is derived in step \ref{alg:Step7Ej}, and steps \ref{alg:Step8E0j} to \ref{alg:StepCheck} remove the nodes that harvest enough energy during the AN period. 
Then it checks the availability of the needed energy in step \ref{step:E0Available} \big($\sum_{i=1}^K (E_{0i}+\tau_i) - 1$ is actually $\sum_{i=1}^K E_{0i} - \tau_0$\big) and updates the interval of the optimal $\phi$ accordingly. The algorithm terminates when the needed energy \big($\sum_{i=1}^K E_{0i}$\big) is equal to the available energy \big($\tau_0 = \tau_0\sum_{i=1}^K a_{0i}$\big). As the needed energy is an increasing continuous function of $\phi$, Algorithm \ref{alg:FixTauE0MaxMin} converges quickly. Finally, we obtain $\bm{\lambda}$ and $\nu$ as they are needed in the computation of the gradient with respect to $\bm{\tau}$.  }
 
\subsection{Proportional Fairness}
The proportional fairness tries to alleviate the unfairness among the nodes without drastically decreasing the sum secrecy throughput. {\color{black}This is achieved by  considering the sum of logarithm of the nodes' secrecy throughput as objective function. The farther a node is away from the BS, the smaller is its throughput. Increasing the throughput of such far away nodes will result in a large gain in the objective function at the cost of consuming more BS resources. The proportional fairness achieves a good trade-off. The PLF problem is formulated as}
\begin{maxi!}
{\bm{E}_{0},\bm{\tau}}{\sum_{i=1}^K\log{(D_i^{sec}(\frac{E_i}{\tau_i}))},}{\label{opt:LogFair}}{}
\addConstraint{\sum_{i=1}^K (E_{0i}+\tau_i) = 1} \label{}
\addConstraint{{E}_{0i},\tau_i \ge 0,  \ \forall i \in \{ 1,...,K\}.}
\end{maxi!}
As proved in Lemma \ref{lem:ConcaveDi}, $D_i^{sec}(\frac{E_i}{\tau_i})$ is concave and the logarithm of a concave function is still concave \cite{Boyd2004}. The objective function is a non-negative  sum of concave functions \big($\log{(D_i^{sec}(\frac{E_i}{\tau_i}))}$'s\big), hence is concave. In addition, the equality constraint is linear, so the PLF problem is convex. With similar arguments to the no fairness and  the max-min fairness cases, we propose a simpler method, which gives us intuition about the solution. We define a vector of slack variables, denoted by $\bm{\psi}=(\psi_1,...,\psi_K)$, and reformulate the problem as 
\begin{maxi!}
{\bm{\psi},\bm{\tau},\bm{E}_{0}}{\sum_{i=1}^K\log(\psi_i),}{\label{opt:SlackAddLog}}{}
\addConstraint{D_i^{sec}(\frac{E_i}{\tau_i}) \ge \psi_i, \ \forall i \in \{ 1,...,K\} }  \label{eq:SlackConst}
\addConstraint{\sum_{i=1}^K (E_{0i}+\tau_i) =1} \label{}
\addConstraint{{E}_{0i},\tau_i \ge 0,  \ \forall i \in \{ 1,...,K\}.}
\end{maxi!}

{\color{black}We show in the proof of Proposition \ref{prop:QPFTauFix} that in the optimal solution, we have $\psi_i^* = D_i^{sec}(\frac{E_i^*}{\tau_i^*})$.} Therefore, the optimal solution of this problem is the same as the original problem. Similar to the no fairness and the max-min fairness cases, we consider the dual problem and use the alternating optimization method to solve this problem. Compared to Section \ref{seq:SectionMaxSumSecrecy} (no fairness case), in each of the fairness problems, we define suitable slack variables and solve a new complex optimization problem for a fixed $\bm{\tau}$. The dual problem can be expressed as 
\begin{maxi!}
{\bm{\lambda},\nu}{\min_{\bm{\psi},\bm{\tau},\bm{E}_{0}}\ \mathcal{L}(\bm{\psi},\bm{\tau},\bm{E}_{0},\bm{\lambda},\nu),}{\label{equ:DualProp}}{}
\addConstraint{\lambda_i,E_{0i},\tau_i \ge 0,\ \forall i \in \{ 1,...,K \}, } 
\end{maxi!}
where $\bm{\lambda} \triangleq (\lambda_1,...,\lambda_K)$ and 
\begin{align}
\mathcal{L}(\bm{\psi},\bm{\tau},\bm{E}_{0},\bm{\lambda},\nu)&=\sum_{i=1}^K 
\big(-\log{(\psi_i)} + \lambda_i (\psi_i - D_i^{sec}(\frac{E_i}{\tau_i}) ) \big) 
\nonumber\\&+\nu(\sum_{i=1}^K (E_{0i} + \tau_i)  -1). \nonumber
\end{align}

{\color{black}We set $\bm{\tau}$ to an initial value and obtain the optimal value of the other variables via algorithm \ref{alg:FixTauE0Prop}. The optimality of this algorithm is proved in Proposition \ref{prop:QPFTauFix}.} Then, the gradient descent method is used to obtain the new $\bm{\tau}$, which is then used to obtain the new optimal of the other variables. Due to the convexity of the above problem, the optimal point is obtained using this iterative method. The algorithm for this part is similar to the previous cases and is omitted for brevity.

\begin{prop}\label{prop:QPFTauFix}
Algorithm \ref{alg:FixTauE0Prop} yields the optimal solution of the problem defined by equation \eqref{equ:DualProp} for a fixed $\bm{\tau}$.
\end{prop}
\begin{proof}
The proof is given in appendix \ref{sec:PropositionLambdaProp}.
\end{proof}

\begin{figure}
\removelatexerror
\begin{algorithm}[H]
\caption{optimal $\bm{E}_{0}$, $\bm{\psi}$, $\bm{\lambda}$ and $\nu$ in the proportional fairness case}\label{alg:FixTauE0Prop}
\begin{algorithmic}[1]
\State \textbf{Initialization:} $\bm{\lambda}=(0,...,0)$, $B=\{1,...,K\}$, $flag = 1$. Set $\nu_{min},\nu_{max} $ such that the optimal $\nu$ lies between them. 
\While{$Er(\nu) > \epsilon$}
\State $\nu = \frac{\nu_{min}+\nu_{max}}{2}$ 
\For{$i \in B$}
\State Solve $D_i^{sec}(\frac{E_i}{\tau_i}) -\frac{\mathcal{B}_i(\frac{E_i}{\tau_i})}{\nu} = 0 $ to find $E_i$ \label{Alg:Step:PropMainEq}
\State $E_{0i} = \max{(0,\frac{E_i}{\mu_i\eta_i P_H}-\sum_{l=1}^{i-1}\tau_l a_{l,i})}$ \label{step:sumchannel}
\EndFor
\State $Er(\nu) \triangleq \sum_{i=1}^K (E_{0i}+\tau_i) - 1$ \label{Alg:Stp:Err}
\If{$Er(\nu)<0$}
\State $\nu_{min} = \nu $
\Else 
\State $\nu_{max} = \nu$
\EndIf
\EndWhile
\State $E_i = \mu_i\eta_i P_H(E_{0i}+\sum_{j=1}^{i-1}\tau_j a_{j,i}),\ \forall i \in \{ 1,...,K \}$
\For{$i \in B$}
\State $\psi_i = D_i^{sec}(\frac{E_i}{\tau_i})$, $ \lambda_i = \frac{1}{D_i^{sec}(\frac{E_i}{\tau_i})}$
\EndFor
\end{algorithmic}
\end{algorithm}
\end{figure}
Algorithm \ref{alg:FixTauE0Prop} is similar to Algorithm \ref{alg:FixTauE0} (no fairness case), and yields the optimal $E_i$'s to maximize its objective function. Similarly, $Er(\nu)$ in step \ref{Alg:Stp:Err} is a decreasing function of $\nu$ and $Er(0) > 0$. 
Hence, it is straightforward to obtain $\nu^*$ such that $Er(\nu^*) =0$. The difference between this case and the no fairness case is in their proposed $E_i$'s. Algorithm \ref{alg:FixTauE0} proposes $E_i$ in step \ref{step:optimalEnergy}
 via solving the equation $\mathcal{B}_i(\frac{E_i}{\tau_i})=\nu$, while Algorithm \ref{alg:FixTauE0Prop} solves the equation $\mathcal{B}_i(\frac{E_i}{\tau_i}) = \nu D_i^{sec}(\frac{E_i}{\tau_i})$ 
to propose $E_i$. {\color{black}In the PLF, we multiply $\nu$ by the throughput. Hence, we decrease the right-hand side of the equation for the farther nodes, which have lower throughput, compared to the nearer ones. This multiplication increases $E_i$ of the farther nodes as $E_i$ is a decreasing function of $\mathcal{B}_i(\frac{E_i}{\tau_i})$. Since the throughput is an increasing function of the harvested energy, the PLF decreases the difference between the throughput of the farther and the nearer nodes, compared to the no fairness case.}
 The equation to obtain the optimal $E_i$ (step \ref{Alg:Step:PropMainEq}) is easy to solve as it is an increasing function of $E_i$ and is less than zero for $E_i= 0$. Finally, we obtain $\bm{\lambda}$ and $\bm{\psi}$ to compute the gradient of the dual problem with respect to $\bm{\tau}$.
\color{black}
   
\section{Numerical Simulations}
In this section, our algorithms are evaluated using various numerical simulations. We consider the scenario of Fig. \ref{fig:Scenario} with a multi-antenna BS at the origin and 4 nodes ($K=4$), which locate at $(1.8,0)$, $(2,\frac{{\pi}}{4})$, $(2.2,\frac{{\pi}}{2})$, and $(5,\frac{{3\pi}}{4})$ polar coordinates. 
We have considered an asymmetric scenario, which is more representative of practical scenarios. We assume using 2.4GHz ISM band, variance of noise $\sigma^2=-100dBm$, and for simplicity, we consider the energy convergence efficiency of all the nodes to be one ($\eta_i = 1$). We consider both i.i.d Rayleigh fading distribution (with parameter one) and Rician distribution (with $K_{Rician}=10$) for the CSI of the channels between the BS and the nodes. For the channels between the nodes, we only consider Rayleigh fading. The path loss exponent is considered to be 3. We first evaluate the sum secrecy throughput performance and then the proposed MMF and PLF fair algorithms    
 
\subsection{Sum Secrecy Throughput}
In this subsection, we compare the SSTM with the following three algorithms: 
\begin{itemize}
\item
\textbf{Uuniform time slotting and weights (UTW)}, the BS uses the same beamforming weights for all the nodes $(a_{i,j}=\frac{1}{K})$, and the duration of all the time slots are equal to each other $(\tau_i = \frac{1}{K+1})$. 
\item 
\textbf{Uniform time slotting (UT)}, {\color{black}the BS applies the first stage of the SSTM (Blinding the non-transmitting nodes) during $\tau_1$ up to $\tau_K$, but all the time slots are equal to each other, and the beamforming weights are the same during $\tau_0$. }
\item 
\textbf{Uniform blinding (UB)}, the BS allocates the same beamforming weights to all the nodes during the information transmission time slots, but the beamforming weights during the non-information transmission time slot and  the duration of the time slots are similar to the proposed algorithm. This algorithm just applies the second stage of the SSTM. It uses a uniform beamforming instead of the beamforming obtained in the first stage of the SSTM.
\end{itemize}
\begin{figure}
\centering
\includegraphics[width = 1\linewidth]{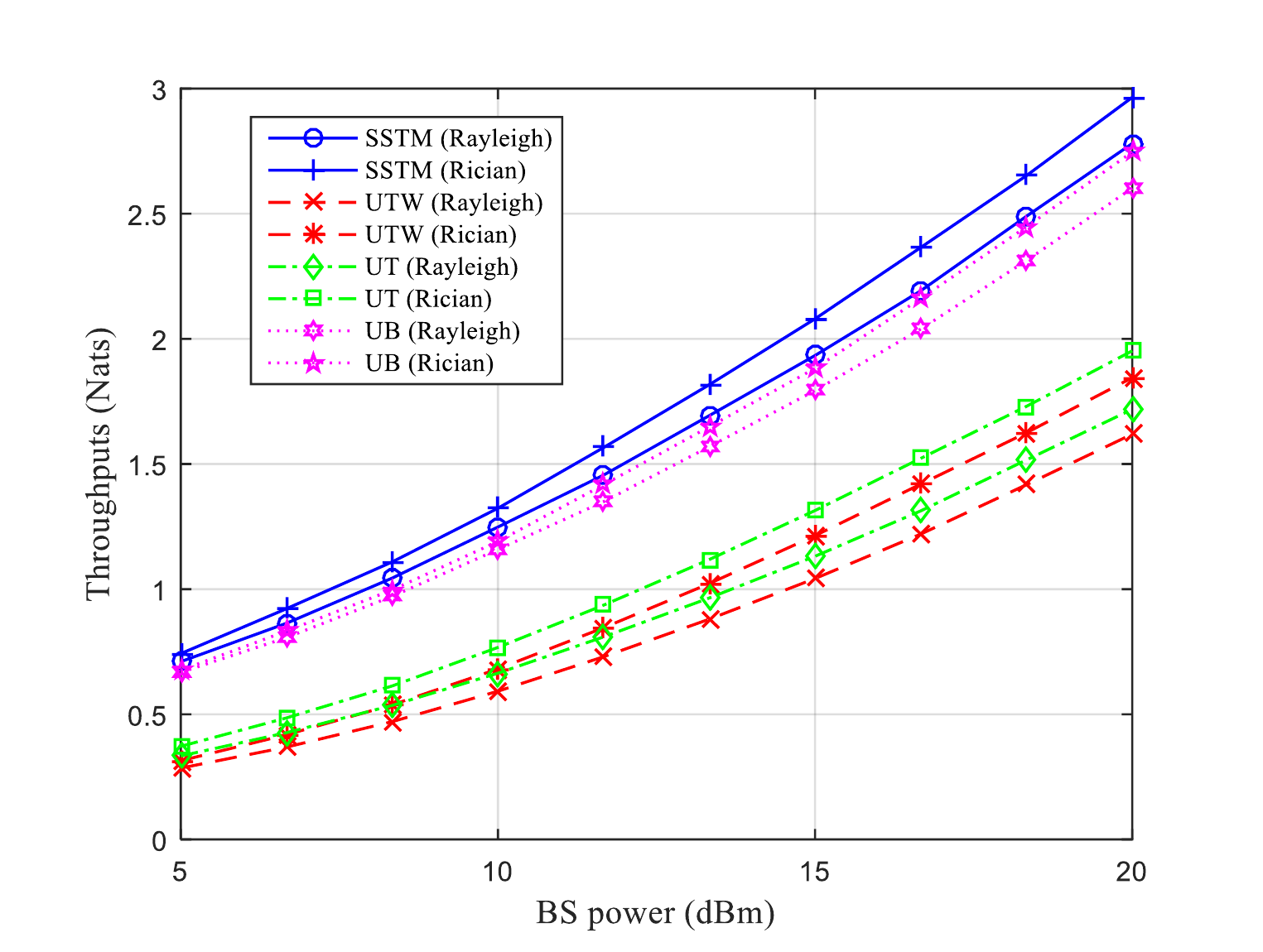}
\caption{Throughput vs BS power with $N = 50$ and $K_{Rician}=10$}
\label{fig:MultiAntennaPh}
\end{figure}

We plot the sum secrecy throughput versus the BS power for Rayleigh and Rician fading in Fig. \ref{fig:MultiAntennaPh}. {\color{black}As it can be seen, the UTW algorithm and the SSTM algorithm have the worst and the best performance, respectively. The UT and the UB algorithms are better than the UTW algorithm, and it shows that applying each stage of the SSTM improves the performance significantly. As the second stage of the SSTM maximizes the sum secrecy throughput directly, the UB has better performance comapred to the UT.
 In all algorithms the performance in Rician model is much better as we have line of sight channel between the BS and the nodes.}

%
%

\color{black}
\subsection{Fairness Among the Nodes}
In this subsection, we compare the performance of the two proposed fair algorithms with the no fairness case. The sum secrecy throughput is presented in Fig. \ref{fig:FairTotalThr}, and as it can be seen the no fairness case  and the max-min fairness case have the highest and the lowest sum secrecy throughput, respectively. The secrecy throughput of each node is presented in Fig. \ref{fig:FairEachNode} to compare the fairness of the three algorithms. For the MMF algorithm, the throughput of all nodes are almost similar to each other. {\color{black}It should be noted that the nodes obtain part of their energies from the AN, and we do not consider the fairness in that stage. This results in a difference between the throughput of the nodes in the MMF algorithm. The MMF and the PLF algorithms increase the throughput of the farther nodes (3 and 4) and decrease the throughput of the nearer nodes (1 and 2).} In Fig. \ref{fig:FairFarthestNode}, we compare the secrecy throughput of node 4, which has the worst channel condition. As it can be seen, the MMF algorithm has the highest secrecy throughput for this node. The MMF, on the otherhand, has the lowest sum secrecy throughput among the three algorithms, as shown in Fig. \ref{fig:FairTotalThr}. For the no fairness case, the secrecy throughput of node 4 is negligible, and therefore we can't see it in Fig. \ref{fig:FairEachNode}

\begin{figure}
\centering
\includegraphics[width = 1\linewidth]{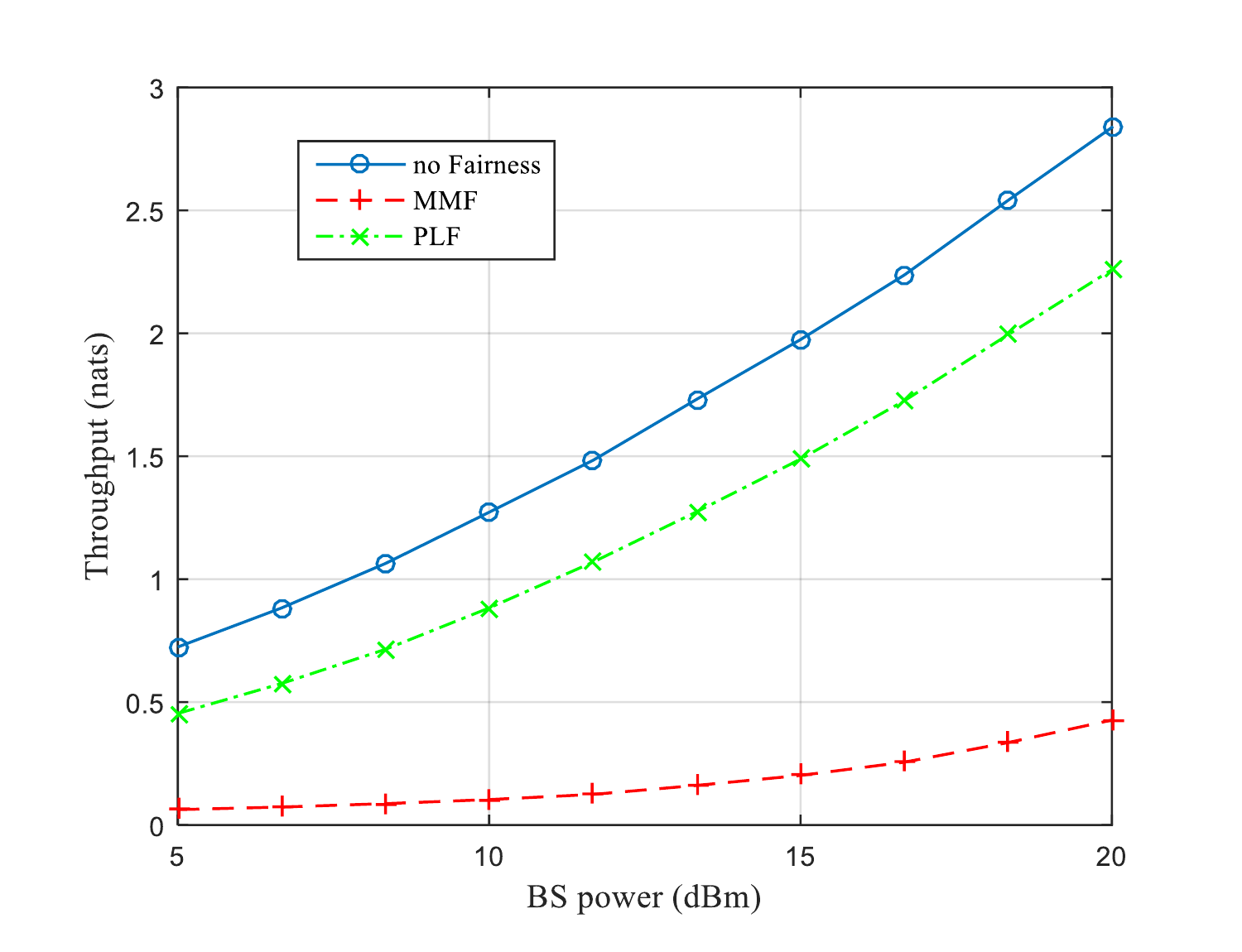}
\caption{Comparison of sum secrecy throughput vs BS power with $N=50$}
\label{fig:FairTotalThr}
\end{figure}
\begin{figure}
\centering
\includegraphics[width = 1\linewidth]{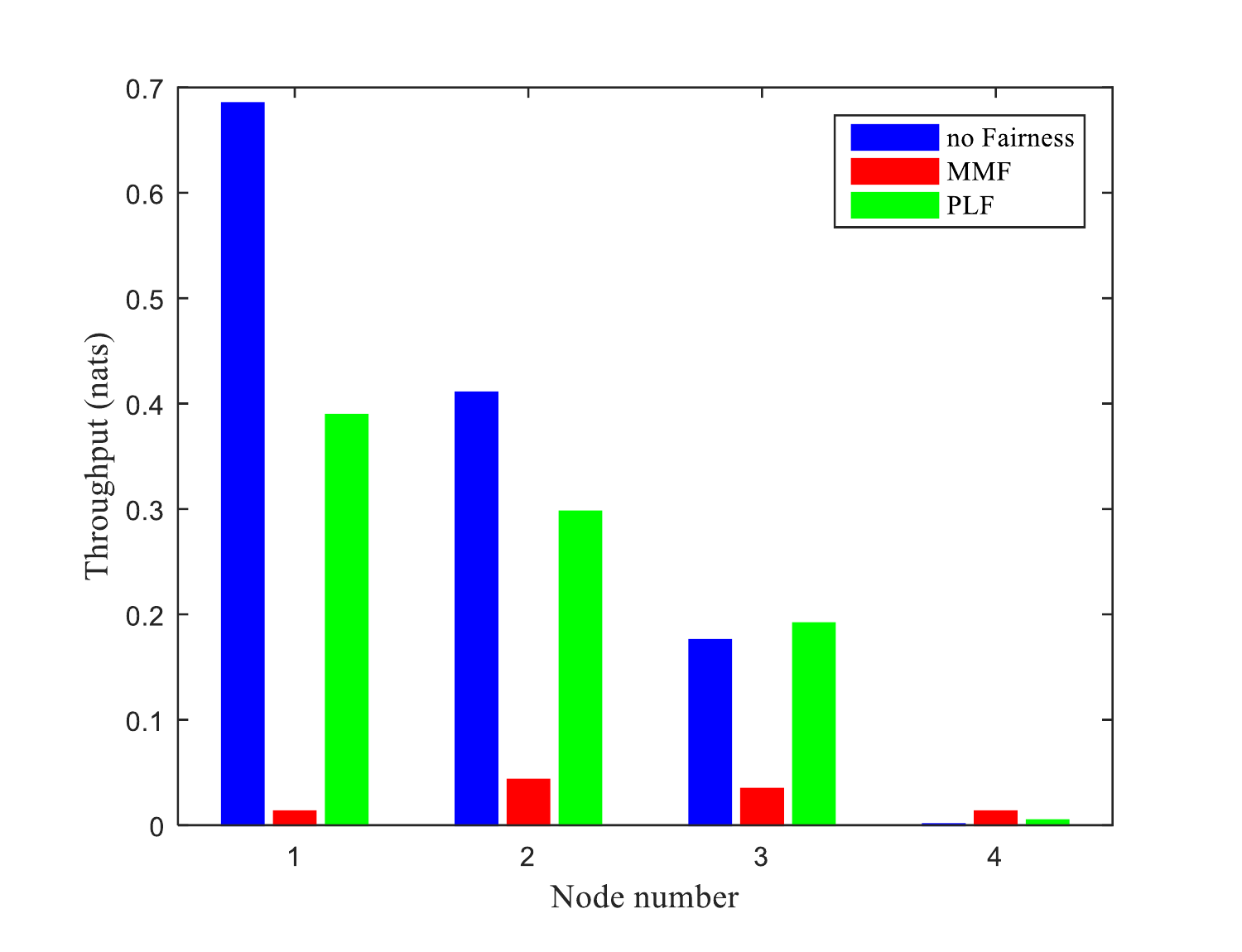}
\caption{Each node throughput with $N = 50$, $P_H = 10\ dBm$}
\label{fig:FairEachNode}
\end{figure}
\begin{figure}
\centering
\includegraphics[width = 1\linewidth]{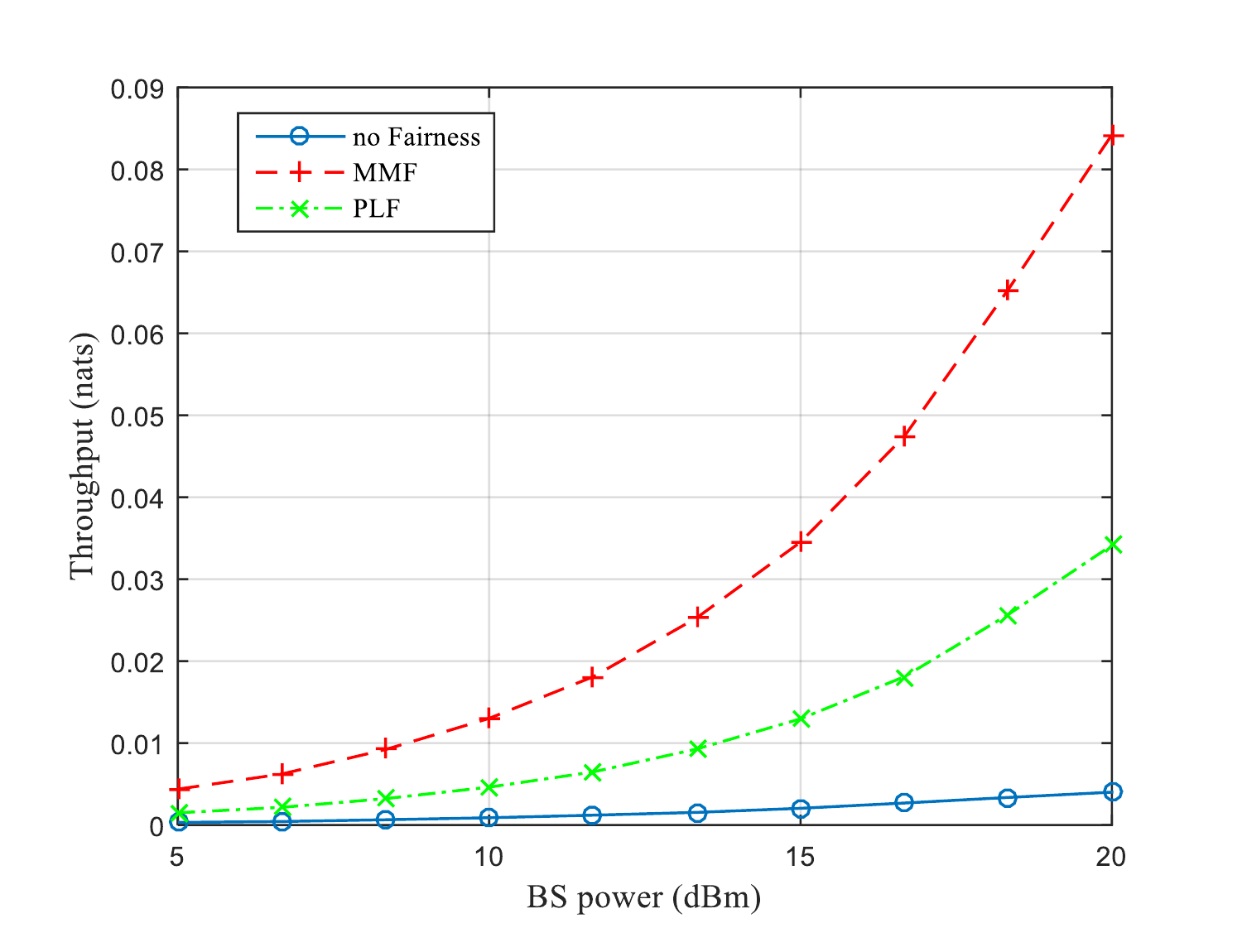}
\caption{Comparison of Throughput of node 4 (the weakest node) vs BS power with $N=50$}
\label{fig:FairFarthestNode}
\end{figure}
\section{Conclusion}
In this paper, the sum secrecy throughput optimization problem in a full-duplex WPCN has been studied. The BS transmits energy all the time, and the nodes transmit their information sequentially to the BS. Each node harvests energy from the beginning until its transmission time, and uses the harvested energy to transmit its information during its allocated time slot. During the transmission of each node, all the other nodes are considered as potential eavesdroppers. The secrecy rate for each node has been derived, and the sum secrecy throughput maximization problem has been formulated. The problem is non-convex, so a suboptimal two stage approach has been proposed to convert the problem into convex optimization problems. In the first stage, the BS optimizes its beamforming vectors to maximize the interference of the non-transmitting nodes. Then, the optimal beamforming vector in the initial non-information transmission time slot and the optimal time slots are obtained. Maximizing the sum secrecy throughput may result in unfairness among the nodes as a consequence of the double-near-far problem. To alleviate this issue, we have proposed the MMF and the PLF algorithms. Through numerical simulations, we first have showed that each of our stages towards optimizing the sum secrecy throughput improves the performance significantly. {\color{black} We have also compared the two proposed fair algorithms with the no fairness case. The PLF algorithm balances between the sum secrecy throughput and the fairness among the nodes, while the MMF algorithm benefits the node with the worst channel condition.}

\color{black}
\appendices
\section{Proof of Proposition \ref{prop:Blinding}}\label{PropositionBeta}
Without loss of generality, we prove subproblem $i \in \{ 1,...,K \}$. The following problem is equivalent to the problem defined by equation \eqref{minEaves}.
\begin{mini!}
{\Phi,\bm{a}_{i}}{\Phi,}{}{}
\addConstraint{\frac{|h_{i,j}|^2}{\sigma^2+\mu_ja_{i,j}P_H}\le \Phi,\ \forall j \in \{ 1 ,...,K \} \backslash \{ i\}} \label{}
\addConstraint{\sum_{j=1}^Ka_{i,j} = 1}
\addConstraint{a_{i,j} \ge 0, \ \forall j \in \{ 1 ,...,K \}.}
\end{mini!}

This problem is convex, and the Slater condition holds, so we can use the Karush-Kuhn-Tucker (KKT) conditions to obtain the optimal solution. The Lagrangian function can be written as
\begin{align}
\mathcal{L} &= \Phi+\sum_{j=1,j\ne i}^K \lambda_j(\frac{|h_{i,j}|^2}{\sigma^2+\mu_ja_{i,j}P_H}-\Phi)\nonumber +\nu(\sum_{j=1}^K a_{i,j}-1).
\end{align}
The KKT conditions can be expressed as
\begin{align}
\lambda_j&(\frac{|h_{i,j}|^2}{\sigma^2+\mu_ja_{i,j}P_H}- \Phi) = 0,\ \forall j \in \{1,...,K\} \backslash \{ i \} \label{SlaterBeta}, \\
 &\frac{\partial \mathcal{L}}{\partial a_{i,j}}=0 \Rightarrow \nu = \frac{P_H \lambda_j \mu_j |h_{i,j}|^2}{(\sigma^2+\mu_ja_{i,j}P_H)^2}, \ \forall a_{i,j}\ne 0 , \label{BetaNeZero}\\
&\sum_{j=1}^K a_{i,j} = 1, \label{betaEquality}\\
&\ \  \sum_{j=1}^K \lambda_{j} = 1.
\end{align}

Equations \eqref{SlaterBeta} and \eqref{BetaNeZero} shows that in the optimal solution, we have either $\frac{|h_{i,j}|^2}{\sigma^2+\mu_ja_{i,j}P_H}= \Phi$ or $a_{i,j} = 0$. First, we assume that all $a_{i,j}$'s are non-zero and obtain them as follows:
\begin{align}
&\Phi= \frac{|h_{i,j}|^2}{\sigma^2+\mu_ja_{i,j}P_H} ,\forall j\ne i, 
\\&\Rightarrow a_{i,l}=\frac{1}{\mu_l P_H}(\frac{|h_{i,l}|^2}{|h_{i,j}|^2}(\sigma^2+a_{i,j}\mu_j P_H)-\sigma^2) , \forall j,l\ne i\\
&\Rightarrow  \sum_{l=1,l\ne i}^K \frac{1}{\mu_l P_H}(\frac{|h_{i,l}|^2}{|h_{i,j}|^2}(\sigma^2+a_{i,j}\mu_j P_H)-\sigma^2)  \stackrel{(b)}{=}1 \\
\Rightarrow &a_{i,j}=\frac{1 + \frac{\sigma^2}{P_H} \big( \sum_{l=1,l\ne i}^K \frac{1}{\mu_l}-\frac{1}{|h_{i,j}|^2}\sum_{l=1,l\ne i}^K \frac{|h_{i,l}|^2}{\mu_l} \big)}{\frac{\mu_j}{|h_{i,j}|^2}\sum_{l=1,l\ne i}^K \frac{|h_{i,l}|^2}{\mu_l}},
\end{align} 
where $(b)$ comes from equation \eqref{betaEquality}. $a_{i,j}$'s must be greater than or equal to zero, so for the computed $a_{i,j}$'s that are less than zero, we must set them to zero and remove them. Then, we compute the new $a_{i,j}$'s and continue this process until we found the optimal $\Phi$. For the optimal $\Phi$, we have $\frac{|h_{i,j}|^2}{\sigma^2}=\Phi$ if $a_{i,j}$ is greater than $0$, and $\frac{|h_{i,j}|^2}{\sigma^2}<\Phi$ if $a_{i,j}$ equals to $0$ (otherwise the dual function will tend to infinity \cite{Boyd2004}). 
 We have $K-1$ eavesdroppers, and in each step either the process terminates, or one of them is removed.
Hence, we have only one eavesdropper left in the $(K-2)^{th}$ step, and its weight is equal to 1. As a result, the process terminates at most in $K-2$ steps.     

\section{Proof of Lemma \ref{lem:ConcaveDi}}
\label{app:LemmaConcave}
\color{black} We prove that $D_i^{sec}(\frac{E_i}{\tau_i})$ is a concave function of $(E_{0i},\tau_1,...,\tau_{i})$, and therefore is a concave function of $\bm{E_0}$ and $\bm{\tau}$. It should be noted that $D_i^{sec}(\frac{E_i}{\tau_i})$ is a perspective function of $D'_i(E_i) \triangleq \log(1+\zeta_i E_i)-\log(1+\xi_i E_i)$. Hence, in order to prove that $D_i^{sec}(\frac{E_i}{\tau_i})$ is a concave function of $(E_{0i},\tau_1,...,\tau_{i})$, we only need to prove that $D'_i(E_i)$ is a concave function of $(E_{0i},\tau_1,...,\tau_{i-1})$ \cite{Boyd2004}. We prove that $D'_i(E_i)$ is a concave function via proving that its Hessian matrix, denoted by $\mathcal{H}_i$, is negative semidefinite. 
\begin{align}
{\mathcal{H}_i} =\alpha_i \bm{\beta}_i\bm{\beta_i}^T
\Rightarrow \forall \bm{v}\in \mathbb{R}^i ,\ \bm{v}^T\mathcal{H}_i\bm{v} = \alpha_i |\bm{v}^T\bm{\beta_i}|^2 ,
\end{align}
where
\begin{align}
\alpha_i &\triangleq (\eta_i\mu_iP_H)^2\bigg(\big(\frac{\xi_i}{1+\xi_iE_i}\big)^2 - \big(\frac{\zeta_i}{1+\zeta_iE_i}\big)^2\bigg)\stackrel{(b)}<0,\\
\bm{\beta}_i &\triangleq (1 , a_{1,i},...,a_{i-1,i}),
\end{align}
and $(b)$ is valid when $\zeta_i>\xi_i$, which is equivalent to a positive secrecy rate ($R_i^{sec}>0$). The Hessian matrix is negative semidefinite as $\alpha_i < 0$.
\color{black}

\section{Proof of Proposition \ref{theo:E0Opt}}
\label{Proposition Lambda}
We assume a fixed $\bm{\tau}$ and write the KKT conditions for the problem defined by equation \eqref{DualProb} as
\begin{align}
&\frac{\partial \mathcal{L}}{\partial E_{0i}} \ge 0 \Rightarrow \mathcal{B}_i(\frac{E_i}{\tau_i}) \le \nu, \ \forall i \in \{1,...,K\} \label{dualNotInf},\\
& \frac{\partial \mathcal{L}}{\partial E_{0i}} = 0 \Rightarrow \mathcal{B}_i(\frac{E_i}{\tau_i}) = \nu,\ \forall E_{0i} \ne 0 \label{E0NotZero}, \\
&\sum_{i=1}^K (E_{0i}+\tau_i)=1, \label{equ:Errrr}
\end{align} 
where $\mathcal{B}_i(\frac{E_i}{\tau_i}) \triangleq\mu_i\eta_iP_H \bigg(\frac{ \zeta_i-\xi_i}{(1+\zeta_i\frac{E_i}{\tau_i})(1+\xi_i\frac{E_i}{\tau_i})} \bigg)$.
We put equation \eqref{dualNotInf} for the purpose that the gradient of the Lagrangian with respect to $E_{0i}$ must be positive, otherwise the dual function will tend to infinity \cite{Boyd2004}. 
We derive $E_{0i}$'s using equation \eqref{E0NotZero} as
\begin{align}
\nu &= \mu_i\eta_iP_H \bigg(\frac{ \zeta_i-\xi_i}{(1+\zeta_i\frac{E_i}{\tau_i})(1+\xi_i\frac{E_i}{\tau_i})} \bigg) \nonumber \\ 
&\Rightarrow  {E_i} = \tau_i\frac{-(\zeta_i +\xi_i)+\sqrt{(\zeta_i -\xi_i)^2+
\frac{4\zeta_i \xi_i \mu_i \eta_iP_H(\zeta_i -\xi_i)}{\nu}}}{2\zeta_i \xi_i}\label{EiLanda} \\
&\Rightarrow E_{0i} = \max\{0,\frac{E_i}{\mu_i\eta_i P_H}-\sum_{j=1}^{i-1}\tau_j a_{j,i}\}, \ \forall i \in \{ 1,...,K \}. \label{E0Landa}
\end{align}

We define $Er(\nu) \triangleq\sum_{i=1}^K (E_{0i}+\tau_i) - 1$, based on equation \eqref{equ:Errrr}. We have $Er(0)>0$, and $Er(\nu_{max})<0,\ \nu_{max}=\max_{i} \mu_i\eta_iP_H(\zeta_i-\xi_i)$.
It is evident from equations \eqref{EiLanda}  and \eqref{E0Landa} that increasing $\nu$ will decrease $E_{i}$'s, and consequently $E_{0i}$'s. Therefore, $Er(\nu)$ is a decreasing continuous function of $\nu$. Hence, we can use a simple bisection method to find $0\le \nu^*\le \nu_{max}$ such that $Er(\nu^*)=0$. This point $(\bm{E_0}^*, \nu^*)$ satisfies all the KKT conditions and therefore is the optimal solution of the problem.     

\section{Proof of Proposition \ref{theo:E0OptMaxMin}}
\label{sec:PropositionLambdaMaxMin}
We write the KKT conditions of the problem defined by equation \eqref{equ:DualMaxMin} for a fixed $\tau$ as
\begin{align}
&\lambda_i(\phi-D_i^{sec}(\frac{E_i}{\tau_i})) = 0,\ \forall i \in \{ 1,...,K \}, \label{equ:MaxMinLambda}\\
&\frac{\partial \mathcal{L}}{\partial E_{0i}} = 0 \Rightarrow \mathcal{B}_i(\frac{E_i}{\tau_i})\lambda_i = \nu,\ \forall E_{0i}\ne 0, \label{equ:MaxMinNu}\\
&\sum_{i=1}^K (E_{0i}+\tau_i) = 1, \label{equ:errMinPhi}\\
&\frac{\partial \mathcal{L}}{\partial \phi} = 0 \Rightarrow \sum_{i=1}^K \lambda_i = 1.  \label{equ:SumLambda}
\end{align}
The slack variable $\phi$ expresses the minimum throughput. Algorithm \ref{alg:FixTauE0MaxMin} in steps \ref{step:BeginingBPhi} to \ref{step:EndBPhi} yields a set for each $\phi$, denoted by $B_\phi$, with the property that for $i \in B_\phi$, we have $(D_i^{sec}(\frac{E_i}{\tau_i}) = \phi, E_{0i} \ne 0)$ and for $i \notin B_\phi$, we have $(D_i^{sec}(\frac{E_i}{\tau_i}) > \phi, E_{0i} = 0)$. This part of the algorithm satisfies equations \eqref{equ:MaxMinLambda} and \eqref{equ:MaxMinNu}. As we increase $\phi$, we increase the minimum throughput of all the nodes, and $E_{0i}$'s increase to reach this minimum throughput. This will increase the error function defined in step \ref{step:E0Available} (based on equation \eqref{equ:errMinPhi}). It is straightforward to show that $Er(0)<0$, and we prove that $Er(\phi_{max})>0$ as follows:
\begin{align}
&\lim_{E_i \to \infty}D_i^{sec}(\frac{E_i}{\tau_i})=\tau_i(\log{\zeta_i}-\log{\xi_i}) ,\ \forall i \in \{ 1,...,K\}\\ 
& \Rightarrow Er(\phi_{max}) = \infty.
\end{align} 
Therefore, with the aid of a simple bisection method, we can find the optimal $0\le\phi^* \le \phi_{max}$ such that $Er(\phi^*) = 0$. This optimal point ($\phi^*$,$B_{\phi^*}$,$\mathbf{E}_{0}^*$) satisfies all the KKT conditions. For the optimal $B_{\phi^*}$, we derive $\lambda_i^*$'s and $\nu^*$ using equation \eqref{equ:MaxMinNu} as follows:
\begin{align}
\lambda_i^* &= \frac{\nu^*}{\mathcal{B}_i(\frac{E_i^*}{\tau_i})},\ \forall i \in B_{\phi^*}\ \Rightarrow \sum_{i\in B_{\phi^*}} \frac{\nu^*}{\mathcal{B}_i(\frac{E_i^*}{\tau_i})}\stackrel{(a)}{=}1 \\
&\Rightarrow \nu^* = \frac{1}{\sum_{i \in B}\frac{1}{\mathcal{B}(\frac{E_i^*}{\tau_i})}}, 
\end{align}
where $(a)$ comes from equation \eqref{equ:SumLambda}. We obtain all the optimal values to compute the gradient with respect to $\bm{\tau}$.

\balance
\section{Proof of Proposition \ref{prop:QPFTauFix}}
\label{sec:PropositionLambdaProp}
The KKT conditions for the optimization problem defined by equation \eqref{equ:DualProp} for a fixed $\tau$ can be expressed as
\begin{align}
\frac{\partial \mathcal{L}}{\partial \psi_{i}} = 0 \Rightarrow &\lambda_i = \frac{1}{\psi_i} \label{equ:ObtainLambda},\\
\lambda_i(\psi_i - D_i^{sec}(\frac{E_i}{\tau_i})) = 0 \Rightarrow &\psi_i = D_i^{sec}(\frac{E_i}{\tau_i}) \label{equ:ObtainPsi},\\
\frac{\partial \mathcal{L}}{\partial E_{0i}} = 0 \Rightarrow &\mathcal{B}_i(\frac{E_i}{\tau_i})\lambda_i = \nu,\ \forall E_{0i}\ne 0 \label{equ:PropNu},\\
&\sum_{i=1}^K (E_{0i}+\tau_i) = 1.
\end{align}

As it can be seen from equation \eqref{equ:ObtainPsi}, for any $\tau$ (including the optimal $\tau^*$) we have $\psi^*_i = D_i^{sec}(\frac{E^*_i}{\tau_i})$. It shows that the optimal solution of the optimization problems defined by equations \eqref{opt:LogFair} and \eqref{opt:SlackAddLog} (which add slack variables) are the same. For each $E_{0i}\ne 0$, we can simplify and rewrite the above equations as
\begin{align}
&\mathcal{B}_i (\frac{E_i}{\tau_i}) - \nu D_i^{sec}(\frac{E_i}{\tau_i}) = 0, \\
&\sum_{i=1}^K (E_{0i}+\tau_i) = 1. \label{equ:errProp}
\end{align}
 We solve the above equations for all the nodes in step \ref{Alg:Step:PropMainEq} and set the negative $E_{0i}$'s to zero in step \ref{step:sumchannel}. The error function defined in step \ref{Alg:Stp:Err} (based on equation \eqref{equ:errProp}) is an increasing function of $E_{0i}$'s. In addition, it is straightforward to show that $E_{0i}$'s are decreasing function of $\nu$. As a result, the error function is a decreasing function of $\nu$. Moreover, since $Er(0)>0$, it is easy to find $\nu_{min}$ and $\nu_{max}$ such that the optimal $\nu^*$ ($Er(\nu^*)=0$)  lies between them. Then, we can obtain $\nu^*$ using a simple bisection method. Finally, we obtain the optimal $\bm{\psi}^*$ and $\bm{\lambda}^*$ using equations \eqref{equ:ObtainLambda} and \eqref{equ:ObtainPsi}, and compute the gradient with respect to $\bm{\tau}$ using all these obtained optimal variables.
 
 \color{black}
 \section*{Acknowledgment}
The first author acknowledges the support of the Agency for Science, Technology and Research (A*STAR) Research Attachment Program (ARAP). The second author acknowledges the support of the A*STAR Industrial Internet of Things (IIoT) Research Program under the RIE2020 IAF-PP Grant A1788a0023.
\color{black}

\end{document}